\documentclass[sn-apa]{sn-jnl}%

\usepackage{graphicx}%
\usepackage{multirow}%
\usepackage{amsmath,amssymb,amsfonts}%
\usepackage{amsthm}%

\usepackage{mathrsfs}%
\usepackage[title]{appendix}%
\usepackage{xcolor}%
\usepackage{textcomp}%
\usepackage{manyfoot}%
\usepackage{booktabs}%
\usepackage{listings}%
\usepackage{amsmath,amssymb, graphicx, url, algorithm2e}

\usepackage{amsfonts}

\usepackage{mathtools}
\usepackage[capitalize,noabbrev]{cleveref}

\usepackage{times}
\usepackage{wrapfig}
\usepackage{mathtools}

\usepackage{mismath}

\usepackage{comment}

\usepackage{bbm}

\newcommand{\mb}{\mathbb}

\newcommand{\wh}{\widehat}

\newcommand{\EE}{\mb{E}}
\newcommand{\RR}{\mb{R}}
\newcommand{\PP}{\mb{P}}

\theoremstyle{thmstyleone}%
\newtheorem{theorem}{Theorem}
\newtheorem{proposition}[theorem]{Proposition}%
\newtheorem{fact}[theorem]{Fact}%
\newtheorem{lemma}[theorem]{Lemma}%

\theoremstyle{thmstyletwo}%
\newtheorem{remark}{Remark}%

\theoremstyle{thmstylethree}%
\newtheorem{definition}{Definition}%

\begin{document}

\title[Effects of Privacy-Inducing Noise on Welfare and Influence of Referendum Systems]{Effects of Privacy-Inducing Noise on Welfare and Influence of Referendum Systems}


\author*[1]{\fnm{Suat} \sur{Evren}\let\thefootnote\relax\footnotetext{Suat Evren, Departments of Mathematics, Economics, and Computer Science, MIT, e-mail: \href{evrenis@mit.edu}{evrenis@mit.edu}.}
}
\author[1]{\fnm{Praneeth} \sur{Vepakomma}}

\affil[1]{ \orgname{Massachusetts Institute of Technology}}




\abstract{Social choice functions help aggregate individual preferences while differentially private mechanisms provide formal privacy guarantees to release answers of queries operating on sensitive data. However, preserving differential privacy requires introducing noise to the system, and therefore may lead to undesired byproducts. Does an increase in the level of differential privacy for releasing the outputs of social choice functions increase or decrease the level of influence and welfare, and at what rate? In this paper, we mainly address this question in more precise terms in a referendum setting with two candidates when the celebrated randomized response mechanism is used. We show that there is an inversely-proportional relation between welfare and privacy, and also influence and privacy.}




\maketitle

\section{Introduction} 
Differential privacy \citep{DMNS06} provides a compelling privacy guarantee to ensure that the outcome of a query over any dataset is substantially not influenced based on the presence or absence of an individual's record. This form of privacy has recently been studied in the context of social choice theory \citep{shang2014, lee2015efficient, DPRankAggregation2017}. A predominant strategy to achieve differential privacy in general even outside the context of social choice theory is to introduce noise or some sort of randomization into the system. One of the issues that has been widely studied in this context of noising is the specific loss of accuracy in releasing the true output of the non-privatized query as caused by increasing levels of privacy preservation. This has been commonly referred to as the \textit{privacy-accuracy} or \textit{privacy-utility trade-off}. Recent work has involved the formalization of other trade-offs such as the trade-off between privacy and fairness \citep{cummings2019compatibility}. In this work, we analyze two other trade-offs. We show that introducing noise to privatize systems that aggregate the preferences of individuals may affect several other fundamental phenomena such as \emph{influence} and \emph{welfare}. 

In this context, does an increase in the level of privacy for releasing the outputs of social choice functions, increase or decrease the level of influence and welfare, and at what rate? In this paper, we mainly address this question in more precise terms and affirmatively answer that this relation is inversely-proportional and shares specific corresponding rates for the popular \emph{$\rho$-correlated randomized response} mechanism of privatization when used in a referendum setting with two candidates.

The noisy  mechanism that we propose and analyze with regard to influence and welfare in this paper is based on a simple coin-flipping perturbation of the input as follows. Let $\rho$ be an exogenous constant in $[0,1]$ and let each original vote made in the ballot take a value of either $1$ or $-1$. The randomized response records each original vote in the ballot as it is with a probability $\rho$ while with probability $1-\rho$, it ignores the original vote and instead records it as either a $1$ or $-1$ with a uniformly random pick.  The resulting probability space is known as \emph{$\rho$-correlated distribution} or \emph{noisy distribution} in the field of analysis of Boolean functions, and it is referred to as the \emph{randomized response} mechanism in the field of differential privacy. \footnote{For a survey of the field of analysis of Boolean functions, see \cite{o2014analysis}. For a survey of the field of differential privacy, see \cite{dwork2014algorithmic}.} We show that this mechanism preserves ordinal relations between the influences of voters for `any' social choice function. Therefore, if Alice had more influence before than Bob, she will still continue to have more influence. 

In the field of analysis of Boolean functions, the notion of the \emph{influence} of a voter is used to measure the power of an individual on the final result of a social choice function. We extend this definition of influence to our probabilistic setting where noise is introduced for privacy, and term this new notion of influence as \emph{probabilistic influence}. 
Similarly, we define \emph{welfare} to address the second issue of capturing how \emph{ideal} a voting rule is. First, we define it for deterministic functions and then we extend this definition to any probabilistic mechanism. We then show the effect of our privacy inducing randomized response on the welfare of the system. In particular, we show that it preserves the ordinal relations between the welfare of voting systems. That is, if a social choice function $f$ had a greater welfare than $g$ in the deterministic setting after the randomized response $M_{\rho}$ is applied based on the exogenous parameter $\rho$, the welfare of $M_{\rho}f$ will continue to be greater than that of $M_{\rho}g$.

In this context, we share precise statements connecting the noising probabilities $\rho$ used in the mechanism $M_\rho$, their effect on level of privacy $\epsilon$ which in turn results in a specific level of influence and welfare expressed in terms of $\rho$. We precisely show that as the level of privacy increases, the welfare and influence happen to decrease at correspondingly specific rates. Arguably, having a higher welfare in a voting system is desirable and therefore we shine light on this new trade-off  between privacy and welfare. In terms of influence, it is questionable whether a decrease in influence with an increase in privacy is desirable or not. We believe it depends on the context, and therefore in this case, we do not refer to it as a trade-off but instead call it a scaling law. However, as we show in \cref{welfare section}, welfare of the society is equal to total influence of the society.

\subsection{Contributions} We contribute towards bridging differential privacy and social choice theory by deriving the following results on the effect of randomized response over influence, welfare and accuracy.
\begin{enumerate}
  \item \textbf{The privacy-influence relationship:} A notion of \textit{influence} is widely used in the analysis of Boolean functions to study social choice functions. We extend the notion of influence to the noisy setting, and call it \textit{probabilistic influence}. We then show a result relating the trade-off between $\rho-$correlated distribution based differential privacy and probabilistic influence. We show that such privatization changes the influence of every single voter by a factor of $\frac{1+\rho^2}{2}$. Thus, the randomized response preserves the ordinal relations between influences of agents while scaling them by a factor depending on $\rho$ while still ensuring their privacy is preserved. 
  \item  \textbf{The Privacy-welfare trade-off:} We define \emph{welfare} $W(f)$ of a social choice function $f$ and extend the definition to probabilistic mechanisms. Then, we show that $W(M_{\rho}f) = \rho \cdot W(f)$, i.e. the randomized response scales the welfare by a factor of $\rho$, whereby preserving the ordinal relations between the welfare of social choice functions.
  \item \textbf{Accuracy analysis:} We restrict the analysis of \emph{accuracy} (or \emph{utility}) of our mechanism to social choice functions, i.e. the functions with range $\{-1,1\}$. We give the accuracy for Dictatorship, Majority, AND, and OR functions. For dictatorship, AND, and OR functions, we provide a theoretical analysis of accuracy. For the Majority function, we give an asymptotic accuracy when $n$ goes to $\infty$ based on the existing results in the literature. We also give an exact analysis of accuracy for the Majority function for small $n$ by using a computational method that involves dynamic programming.
\end{enumerate}
\subsection{Organization} 
The rest of the paper is organized as follows. In \cref{motivation}, we provide further motivation and background. In \cref{RR and privacy}, we formally describe the differentially private randomized response mechanism. In \cref{probabilistic influence section}, we introduce the notion of probabilistic influence, and give one of our main results that influence scales down by the same constant for every individual. In \cref{welfare section}, we introduce the concept of welfare for general probabilistic mechanisms, and analyze it for randomized response. We shed light into the connection between influence and welfare, and give our second main result that randomized response scales down welfare by the same factor for any given social choice function. In \cref{accuracy analysis section}, we provide an analysis of the accuracy for the randomized response mechanism. In \cref{conclusion}, we discuss the possible future work and the limitations of this paper, and we conclude. Some preliminaries from social choice theory are provided in \cref{social choice theory preliminaries}. All of the proofs are relegated to \cref{appendix: proofs}.

\section{Motivation} \label{motivation}
To intuitively expand on the potential relation between privacy and influence, consider an instance where it might be the case that introduction of noise for the sake of obtaining privacy results in undesired shifts of the power held by different individuals in deciding the finally selected outcome. For example, say that a voter Alice would have had more impact on the outcome than Bob in a case where there is no privatization. It could as well be the case that the power balance shifts to Bob having more impact than Alice after a privacy-inducing noise is introduced. We conclusively show that this cannot be the case as the influence scales down for every voter with the increasing level of privacy by the same constant in the case of the popular randomized response privacy mechanism.

Secondly, regarding the potential relation between privacy and welfare, consider an instance where it may be the case that upon introduction of noise, the chosen social choice function that was originally used to aggregate the individual preferences into a final outcome ends up not being ideal anymore. Hence, it may instead be desirable to switch to another social choice function. For example, suppose that a system uses the majority function to decide which one of the two candidates is elected in the deterministic case. However, the majority function could be severely affected in some instances upon introduction of noise, and another function could end up being a \emph{better} choice. We show that as the privacy increases in the randomized response mechanism, the welfare of each social choice function scales down proportionally. This implies that if a function is a welfare maximizer before introducing noise, it still is a welfare maximizer after the introduction of the noisy mechanism. These two results are especially useful, as they imply that the designers of the initial deterministic social choice mechanism do not have to be concerned about whether their design is robust to the introduction of noise in terms of influence and welfare. 



We now discuss the work that has been done regarding influence and welfare in the context of social choice theory. Influences have long been studied in discrete Fourier analysis and theoretical computer science. The notion of influence was first introduced by \cite{penrose46} and it was first systematically studied by \cite{influenceBOL85}. Some other novel works related to influences in the context of social choice theory include, but are not limited to, KKL Theorem \citep{KKL} and the Majority is Stablest Theorem \citep{MajorityIsStablest2010}. We extend the notion of influence to the noisy setting and call it \emph{probabilistic influence}, and prove a direct linear relation between deterministic influence and probabilistic influence.

The question of the ideal voting rule has long been a matter of discussion in social choice theory. When there are only two candidates, the answer is relatively simple as the majority function seems to be the most ideal voting rule. \cite{May1952} showed that majority is the only social choice function that is anonymous and monotone among all two-candidate voting rules. For more than two candidates, different objectives may result in different voting rules, or even in impossibility results \citep{arrow1950difficulty, arrow2012social, Guildbau, GK68, GS73}. \cite{Hillinger2005} studies various aspects of utilitarian voting. Finding the best function in computationally efficient ways has been studied in the recent field of computational social choice theory. The works of \cite{Mandal2019}, and \cite{Mandal2020} aim to maximize welfare given each voter's utility for candidates in a `distortion framework' in which there is a lack of information about voter's utilities. In that framework, a typical approach is to attempt to maximize the worst-case objective. 

To the best of our knowledge, a definition of welfare that is closest to ours is the one given by O'Donnell (\citeyear{o2014analysis}, page 51). Although they do not explicitly define welfare of a social choice function, there is a linear relation between the expected value of their objective function and the way we define welfare. However, our main conceptual contribution is that our definitions are extend to hold for probabilistic mechanisms and we analyze the effects of privacy on influence and welfare. \cite{o2014analysis} proved that among all two-candidate voting rules, majority is the unique maximizer of welfare, whose proof is essentially based on \cite{Titsworth}. Our main objective is not to find the function that maximizes the welfare; that is rather a simple question. In fact, we show that majority is the unique welfare maximizer as well in an almost identical way to O'Donnell. The primary motivation of the paper is to show that if a voting rule is better in the deterministic setting, it is still better after the privacy-inducing noise is introduced.

\section{Model: Randomized Response and Privacy Guarantee} \label{RR and privacy}
There are three main reasons as to why we chose the randomized response as the privacy-preserving mechanism to focus our attention. First, it is simple, in addition to being one of the earliest, and yet one of the most popularly used privacy-preserving mechanisms to date, be it in the classic form or as a variant of it. As an example, RAPPOR \citep{erlingsson2014rappor} is a recent popular real-world use-case of randomized response, otherwise classically used a few decades ago \citep{warner1965randomized, mangat1994improved}. Second, the mechanism is based on perturbations of the input which allows it to be applied to `any' social choice function. This enables us to talk about the ordinal relations between the welfare of potential social choice functions before and after the mechanism is applied. Third, $\rho$-correlated distributions are well studied in mathematical social choice theory \citep{o2014analysis}.

Our randomized mechanism is an input-perturbing mechanism. That is, the mechanism introduces noise to the votes in the ballot so that one can use any social function afterward, yet the same privacy guarantee will continue to hold due to the post-processing property 
\cite{dwork2009differential} of differential privacy. Randomized response introduces noise by utilizing a simple coin-flip scheme that is based on the following distribution that is widely used in the analysis of Boolean functions.  \\

\begin{definition} Let $\rho \in [0,1]$ and $x \in \{-1,1\}^n$ be fixed. $y$ is called \emph{$\rho$-correlated with $x$} if for every $i \in [n]$, $y_i = x_i$ with probability $\rho$ and uniformly distributed with probability $1 - \rho$, and it is denoted by $y \sim N_{\rho}x$. \\
\end{definition}

Note the symmetry in the definition of $\rho$-correlation. We formalize this symmetry in the following fact, which we will often use in the proofs of our results. \\

\begin{fact} \label{switching lemma} 
$x \sim \{-1,1\}^n, y \sim N_{\rho}x$ if and only if $y \sim \{-1,1\}^n, x \sim N_{\rho}y$. If $x \sim \{-1,1\}^n, y \sim N_{\rho}x$, we say $(x,y)$ is a \emph{$\rho$-correlated uniformly random pair}. \\
\end{fact}

In the literature, $\rho$-correlated distribution is sometimes referred to as \emph{noisy distribution}. A famous analogy for this definition is as follows. Suppose the votes are recorded by a \emph{noisy} machine. That is, the machine records each ballot correctly with probability $\rho$, and blurs the ballot with probability $1-\rho$ and instead records it at uniform random. As a result, the vote gets misrecorded with probability $(1-\rho)/2$. In fact, our mechanism corresponds to this noisy machine. Hence, we will call it by the generic name \emph{randomized response}, or \emph{$\rho$-correlated randomized response} when we need to specify $\rho$ and denote a mechanism that applies it by $M_\rho$ as defined below. \footnote{Note the subtle distinction between $M_\rho$ and $N_\rho$.} It is worth noting that \emph{$\rho$-correlated randomized response} is in essence just like \emph{randomized response} \citep{warner1965randomized}, a classic scheme that inspired several privacy mechanisms. \\

\begin{definition}
Let $f:\{-1,1\}^n \to \RR$ be any function. For every $x \in \{-1,1\}^n$, the \emph{randomized response} $M_{\rho}f(x)$ outputs $f(y)$ where $y \sim N_{\rho}x$. \\
\end{definition}

Now that we formally defined the randomized response mechanism, we can give the formal definition of differential privacy in our context.\\

\begin{definition}[$\epsilon$-Differential Privacy \cite{dwork2014algorithmic}] 
A randomized voting mechanism $\mathcal{A}:\{-1,1\}^n \to \{-1,1\}$ is $\epsilon$-differentially private if for all pair of neighboring voting profiles $\mathbf{x, x'}\in \{-1,1\}^n$ that differ in exactly one bit and for all $\mathbf{s} \in \{-1,1\}$,
$$
\Pr[\mathcal{A}(\mathbf{x}) =  \mathbf{s}] \leq e^\epsilon \Pr[\mathcal{A}(\mathbf{x'}) = \mathbf{s}] 
$$
\end{definition}

The above definition of differential privacy is specific to our context. For the general definition of differential privacy and a broad survey of the field, see \cite{dwork2014algorithmic}. The randomized response mechanism preserves $\varepsilon$-differential privacy. The following result holds for any Boolean function $f$. \\

\begin{proposition} \label{privacy}
For any $\rho \in [0,1]$, randomized response $M_{\rho}f$ preserves $\log({\frac{1+\rho}{1-\rho}})$-differential privacy regardless of the function $f: \{-1,1\}^n \to \RR$.
(or, ($\varepsilon$,0)-differential privacy when $\rho \leq 1 - \frac{2}{\exp(\varepsilon) +1}$). \\
\end{proposition}

\begin{proof}
Proof is relegated to Appendix \ref{proof of RR}.
 \end{proof}
 
 \begin{remark}\label{privacy remark}
 The equality case is satisfied if $f$ is a dictatorship, which implies that the bound $\log({\frac{1+\rho}{1-\rho}})$ is tight. That is, when $f$ is a dictatorship, $M_{\rho}f$ is not $\varepsilon$-differentially private for any $\varepsilon < \log({\frac{1+\rho}{1-\rho}})$. In fact, it can be shown that a social choice function $f$ satisfies the equality case if and only if there is a triple $(r,b,i)$ where $r \in \RR, b \in \{-1,1\}, i \in [n]$ such that $\emptyset \neq \{z \in \{-1,1\}^n| f(z) =r\} \subseteq \{z \in \{-1,1\}^n| z_i = b\}$. \\
 \end{remark} 

 The reason our mechanism preserves differential privacy for any Boolean function $f$ is that the mechanism is input-perturbing. In this sense, we could instead present the mechanism as $M_{\rho} : \{-1,1\}^n \to \{-1,1\}^n$ and write $f \circ M_{\rho}$ instead of $M_{\rho}f$. Then we could prove the analogous version of Proposition \ref{privacy}, and by using the post-processing property of differential privacy, we would again obtain Proposition \ref{privacy}. In fact, one can see that in the proof, we also prove the post-processing property, seemingly for no reason. However, the reason we choose to give the mechanism altogether after post-processing with $f$ is to make the all equality cases in the above remark apparent. Once post-processing is applied black-box, whether the privacy result is robust is not clear anymore. For example, consider any constant function $f$, e.g. $f(x) = 1$ for any $x \in \{-1,1\}^n$. In this case, $M_{\rho}f$ is not only $\log(\frac{1+\rho}{1-\rho})$-differentially private but $0$-differentially private.

\section{Probabilistic Influence} \label{probabilistic influence section}

Influence of a voter is a notion that is used to measure the power of an individual on a deterministic social choice function. Influences of Boolean functions have long been studied in computer science and the field of analysis of Boolean functions starting with \cite{influenceBOL85}. The \textit{influence} of a voter in a voting system is defined to be the probability of the change in outcome when the voter changes their vote \emph{ceteris paribus}. For example, in the case of a dictatorship, the dictator has influence $1$ while every other voter has influence $0$. In the majority function with $n=2k+1$ voters, each voter's influence is the same and equal to ${2k \choose k} /2^{2k}$.

We use  $x_{i \to 1} = (x_1, \cdots ,x_{i-1},1,x_{i+1}, \cdots, x_n)$ to denote the case where the $i$-th voter chooses to vote for $1$, and every other voter follows $x$. Similarly, we denote the alternate case where the $i$-th voter chooses to vote for $-1$ and every other voter follows $x$ by $x_{i \to -1} = (x_1, \cdots ,x_{i-1},-1,x_{i+1}, \cdots, x_n)$.  Using this notation,  influence in the deterministic setting is defined as follows. \\

\begin{definition}
 For $f:\{-1,1\}^n \to \{-1,1\}$, the influence of elector $i$ is defined as $$I_i[f] = \PP_{x \in \{-1,1\}^n}[f(x_{i \to 1}) \neq f(x_{ i \to -1})]$$ 
 The total influence of the function $f$ is defined to be 
\[ I[f] = \sum_{i = 1}^{n} I_i[f] \]
\end{definition}

A similar notion can be introduced in the probabilistic setting where the randomized response $M_{\rho}f(x)$ is applied. To do so, we consider the case where everybody casts their votes, following which $M_{\rho}f(x)$ is applied and the voter $i$ changes their vote. That is, we leave all the noisy versions of the votes cast by everyone as is except for the elector $i$'s vote. For this particular vote, we re-run the randomized response on coordinate $i$. The probability of result being different is called the \textit{probabilistic influence} of coordinate $i$. We now introduce the formal definition of the proposed probabilistic influence, which applies not only to social choice functions with range $\left\{-1,1\right\}$ but to all Boolean functions with range in $\mathbb{R}$ as follows. In the notation of the following definition, $y_i \sim N_{\rho}(1)$ refers to the case where voter $i$ chooses to vote for $1$ while $z_i \sim N_{\rho}(-1)$ refers to the case where voter $i$ chooses to vote for $-1$.\\

\begin{definition}
Let $f: \{-1,1\}^n \to \RR$ and the \textit{probabilistic influence} of coordinate $i$ in a mechanism $M_{\rho}f(x)$ is defined as 
\[ I_i[M_{\rho}f] = \EE_{x \sim \{-1,1\}^n, \forall j \neq i \text{ } z_j=y_j=x_j, y_i \sim N_{\rho}(1), z_i \sim N_{\rho}(-1)}[\left(\frac{f(y) - f(z)}{2}\right)^2] \]
The total influence of the mechanism $M_{\rho}f$ is defined to be 
\[ I[M_{\rho}f] = \sum_{i = 1}^{n} I_i[M_{\rho}f] \]
\end{definition}

We showed in Proposition \ref{privacy} that our probabilistic voting mechanism preserves $\varepsilon$-differential privacy. 
Inducing such privacy requires probabilistic mechanisms as opposed to using deterministic functions. For example, in the majority voting with $2k+1$ voters, if the votes are split $k$ to $k+1$, then changing only one bit in the input may change the outcome of the voting mechanism. Thus, it is not differentially private. Similarly, no deterministic Boolean function can preserve differential privacy unless it is a constant function. 

On the other hand, introducing noise may cause several issues in the voting system, one of which is the accuracy of the mechanism, which we will discuss in more detail in \cref{accuracy analysis section}. Another possible issue is that when noise is introduced, we might be altering the voting system in favor of a particular voter. For example, voter $A$ might have more influence relative to voter $B$ in the system now even if that was not the case before. For symmetric social choice functions, it is natural to expect that the randomized response mechanism would have the same effect for any voter since the noise is also symmetric. However, it is not as trivial for arbitrary social choice functions. Yet, we show that each voter's probabilistic influence is proportional to her influence in the deterministic setting. Therefore, the randomized response preserves the ordinal relations between influences of the voters regardless of the original social choice function being used. In other words, if voter $A$ had greater influence than another voter $B$, she will still have a greater influence on the system after the noise is introduced.  \\



\begin{theorem} \label{influence theorem}
Let $\rho \in [0,1]$ be any real number and $f: \{-1,1\}^n \to \RR$ be any function. For every $i \in [n]$, $I_i[M_{\rho}f] =\frac{1 + \rho^2}{2} I_i[f]$.
\end{theorem}

\begin{proof}
    Proof is relegated to \cref{appendix proof of influence theorem}.
\end{proof}

\section{Welfare} \label{welfare section}

In this section, we introduce a formal definition of \emph{welfare} of social choice functions. Then we extend this definition to probabilistic mechanisms, and we show that the randomized response preserves the ordinal relations between the welfare of social choice functions. 

\subsection{Welfare of Deterministic Voting Systems}

\cite{rousseau} argues in his \emph{Social Contract} that an ideal voting rule should maximize the number of votes that agree with the outcome. For a more comprehensive discussion on this, see \cite{S08}. \cite{o2014analysis} proves that the majority function is the unique ideal function based on Rousseau's perception of the ideal voting rule without formally introducing welfare. Perhaps, when he proved this result, he had some form of welfare in his mind, especially because he uses the letter $w$ to denote the number of votes that agrees with the outcome. In this section, we will formally define welfare, which will be slightly different than what the $w$ notation of O'Donnell describes. In particular, we define \textit{welfare} of a social choice function $f: \{-1,1\}^n \to \{-1,1\}$ as the average difference between the number of votes that agree with the outcome and the number of votes that do not agree with the outcome under the impartial culture assumption. \\

\begin{definition}
Let $f: \{-1,1\}^n \to \{-1,1\}$ and $x \in \{-1,1\}^n $, and let $w_x(f) =|\{i ; x_i = f(x)\}| - |\{i ; x_i \neq f(x)\}|$. Welfare of the social choice function $f$ is defined to be
\[W(f) = \EE_x[w_x(f)].\]
\end{definition} 

We can still prove that the majority function is the unique maximizer of welfare when $n$ is odd by using a similar method as in the proof of Theorem 2.33 in \cite{o2014analysis}. \\

\begin{proposition} \label{majority is welfare maximizer}
When $n$ is odd, the unique maximizer of $W(f)$ is the majority function. 
\end{proposition} 

\begin{proof}
    Proof is relegated to \ref{proof majority welfare maximizer}.
\end{proof}

Without further assessment, it is not possible to say whether we prefer total influence to be larger or smaller for the welfare of society in a voting system. As we show in the following result, if the social choice function is monotone - that is if a voter changes her vote in favor of a candidate, then this candidate should be weakly better off – then these two notions collide with each other. This result has implications beyond being a simple identity, making the case that if we want to achieve a greater social welfare while adhering to monotone social choice functions, we must choose a function with a greater total influence. \\

\begin{proposition} \label{W = I}
Let $f$ be any monotone social choice function $f:\{-1,1\}^n \to \{-1,1\}$. Then, $W(f) = I[f]$.
\end{proposition} 

\begin{proof}
Proof is relegated to \cref{proof of W = I}. 
\end{proof} 

\subsection{Welfare of Noisy Mechanisms}

To capture the same notion for the probabilistic functions as well, we similarly define welfare of a randomized mechanism applied on a social choice function as follows. Note that the following definition is not only for the randomized response $M_{\rho}$, but any mechanism defined on social choice functions. \\

\begin{definition}
Let $f: \{-1,1\}^n \to \{-1,1\}$, $x \in \{-1,1\}^n $, and $M$ be any mechanism. Let $w_x(Mf) =|\{i ; x_i = Mf(x)\}| - |\{i ; x_i \neq Mf(x)\}|$. Welfare of the mechanism $M$ with the social choice function $f$ is defined to be 
\[W(Mf) = \EE_{x, M}[w_x(Mf)]\]
where the expectation is both over $x$ and the mechanism $M$. \\
\end{definition}

We showed in \cref{influence theorem} that although introducing $\rho$-correlated noise in a voting system has negative effects on influences, it does not provide an unfair advantage to any agent. Another possible undesired byproduct of a randomized mechanism could be that the effect of randomization on the welfare of a particular voting system is more severe compared to the other voting systems. For example, we showed in \cref{majority is welfare maximizer} that the majority function is the unique welfare maximizer. It could be the case that after we introduce noise, it is more likely in the majority function that the outcome will change. Within this context, the following result implies that every voting system is equally affected by the input-perturbing randomized response mechanism. Therefore the randomized response preserves the ordinal relations between the welfare of two-candidate voting systems. \\

\begin{theorem} \label{welfare theorem}
Let $f$ be any social choice function $f : \{-1,1\}^n \to \{-1,1\}$. Then, $W(M_{\rho}f) = \rho \cdot W(f)$. 
\end{theorem}

\begin{proof}
    Proof is relegated to \cref{proof of welfare theorem}
\end{proof}

This result, together with \cref{majority is welfare maximizer}, implies that the majority function is the unique welfare maximizer also after the noise is introduced by applying the randomized response mechanism.

\section{Accuracy Analysis} \label{accuracy analysis section}
There is one significant drawback of the randomized response privatization mechanism in consideration. It is hard to analyze the accuracy of releasing the output of social choice functions upon privatizing it with the randomized response. Although our main objective in this work is not about the analysis of accuracy, we will dedicate a section to the analysis of accuracy for the sake of completeness. As a first pass, we easily find a \emph{generic} lower-bound on accuracy of the randomized response, but it ends up to be so low that it makes it redundant. Therefore, we restrict our analysis to \emph{specific} social choice functions. We theoretically provide results on accuracy for dictatorship, AND, and OR functions.\footnote{For formal definitions of these widely known social choice functions, see Appendix \ref{social choice theory preliminaries}.} In addition, we give a tight lower bound as well as an upper bound for the accuracy of majority function. We also give an algorithm to calculate exact accuracy of majority function by using dynamic programming via memoization. The dynamic programming approach avoids the need to make calculations over every entry in the power-set and instead is much more efficient, while still resulting in an exact solution for computing the accuracy. Our definition of accuracy is in-fact the average of accuracy under the impartial culture assumption. That is,
\[ Acc(M_{\rho}f) = \PP_{\substack{x \sim \{-1,1\}^n \\ M_{\rho}}}[M_{\rho}f(x) = f(x)].\]

Now, we define the \emph{noise operator}, also referred to as the noisy Markov operator, which is a linear operator on the set of Boolean functions. This operator will be useful for accuracy calculations. \\

\begin{definition} \label{noiseOp}For any $\rho \in [0,1]$, the noise operator $T_{\rho}$ is the linear operator on the set of functions $f:\{-1,1\} \to \RR$ defined by 
\[ T_{\rho}f(x) = \EE_{y \sim N_{\rho}x}[f(y)]. \]
\end{definition}

Before we start our analysis, let us also give the definition of \emph{noise stability}. \\ 

 \begin{definition}
 For any $\rho \in [0,1]$ and $f: \{-1,1\}^n \to \RR$, $\rho$-correlated noise stability of $f$ is given by
 $$Stab_{\rho}(f)=\EE_{\substack{x \sim \{-1,1\}^n \\ y \sim N_{\rho}(x)}} [f(x) \cdot f(y)]$$    
 
 \end{definition}

 There is a linear relation between the noise stability of a function and accuracy of the randomized response on this function. Note that $M_{\rho}f(x) \cdot f(x) = 1 $ if $M_{\rho}f(x) = f(x) $, $M_{\rho}f(x) \cdot f(x) = -1 $ otherwise. Thus, 
 \begin{equation}\label{Stab-Acc}
 2 \cdot Acc(M_{\rho}f) -1
 = 2 \cdot \PP_{\substack{x \sim \{-1,1\}^n \\ y \sim N_{\rho}(x)}} [f(y) = f(x)] - 1
 = \EE_{\substack{x \sim \{-1,1\}^n \\ y \sim N_{\rho}(x)}} [f(y) \cdot f(x)] 
 = Stab_{\rho}(f).
 \end{equation}
 
 Also, note that
 \begin{equation} \label{Stab-Markov}
 Stab_{\rho}(f)=\EE_{\substack{x \sim \{-1,1\}^n \\ y \sim N_{\rho}(x)}} [f(x) \cdot f(y)] = \EE_{x \sim \{-1,1\}^n} [f(x) T_{\rho}f(x)].
 \end{equation}
 
 The reason we feel the need to write accuracy in terms of stability is that in the field of Analysis of Boolean functions most results are given in terms of stability for convenience. Yet, we use stability explicitly only when we analyze the accuracy of the majority function.
 
\subsection{Majority}

In this section, we will give the asymptotic accuracy for $Maj_{n}$ function where $n$ is an odd number that goes to infinity. \\

\begin{lemma}[Proposition 10, \cite{ODonnel2004}] 
For any $\rho \in [0,1)$, $Stab_{\rho}[Maj_n]$ is a decreasing function of $n$ where $n$ is an odd number, with
\[ \frac{2}{\pi} \arcsin(\rho) \leq Stab_{\rho}[Maj_n] \leq \frac{2}{\pi} \arcsin(\rho) + O(\frac{1}{\sqrt{1-\rho^2}\sqrt{n}}).\]
\end{lemma}



By using the fact that accuracy is equal to $\frac{1}{2} + \frac{1}{2}Stab_{\rho}(f)$ due to \cref{Stab-Acc}, we get that
\begin{equation} \label{asymptotic accuracy of majority}
    \frac{1}{2} + \frac{1}{\pi} \arcsin(\rho) \leq Acc[M_{\rho}(Maj_n)] \leq \frac{1}{2} + \frac{1}{\pi} \arcsin(\rho) + O(\frac{1}{\sqrt{1-\rho^2}\sqrt{n}}).
\end{equation}

Despite this fact being quite useful, there is no convenient way to calculate the exact value of accuracy of the randomized response on Majority function. Hence, we compute it using dynamic programming via memoization in the following section.

\subsubsection{Algorithm to compute the exact accuracy for small $n$}We now provide a dynamic programming algorithm with memoization to compute the accuracy of the randomized response. In particular, we give the algorithm to calculate the accuracy of the \textit{threshold functions}, that are of the form
\[f_{\theta}(x) = 
\begin{cases}
1 & \text{if } \sum_{i \in [n]}x_i > \theta \\
-1 & \text{if } \sum_{i \in [n]}x_i  \leq \theta
\end{cases}\]
Note that $Maj_n = f_0(\cdot)$ where it takes care of ties by considering them as if $-1$ is the winner. In general, we work with the odd number of voters when we talk about the majority function. But as a simple trick, we will compute it for any $n$ based on the generic definition of the threshold function we gave above since it makes the algorithm less involved.

We now state the noise operator $T_{\rho} f_{\theta_{0}}(x)$ as introduced in \cref{noiseOp} when applied to threshold functions as a way to quantify the expected accuracy as
 \[T_{\rho} f_{\theta_{0}}(x) = \mathbb{E}_{y \sim N_{\rho}x}\left[{1}\left(y_{1}+\ldots y_{n}>\theta_{0}\right)\right].\]

Let $x_{-n}$ denote $x$ without the last bit. In particular, if $x = (x_1, x_2, \cdots, x_{n-1}, x_n)$, then $x_{-n} = (x_1, x_2, \cdots, x_{n-1})$. Note that $x_{-n} \in \{-1,1\}^{n-1}$ while $x \in \{-1,1\}^n$. Then, the stability can be defined using two calls of recursion as follows
$$\label{eq:first}{
T_{\rho} f_{\theta_{0}}{(x)}=\frac{1+\rho}{2} T_{\rho} f_{\theta_{0}-x_{n}}\left(x_{-n}\right)+\frac{1-\rho}{2} T_{\rho} f_{\theta_{0}+x_{n}}\left(x_{-n}\right)
}$$
That is because
\begin{align*}
    \mathbb{E}_{y \sim N_{\rho}x}&\left[1\left(y_{1}+\cdots + y_{n} > \theta_{0}\right)\right] \\
    &=\mathbb{E}_{y_{n} \sim N_{\rho}x_{n}}\left[\mathbb{E}_{y_{-n} \sim N_{\rho}x_{-n}}\left[\mathbbm{1}\left(y_{1}+\cdots + y_{n-1} > \theta_{0} - y_n \right) \mid y_{n}\right]\right] \\
    &= \frac{1+\rho}{2} \underset{y_{-n}\sim N_{\rho}\left(x_{-n}\right)}{ \mathbb{E}} \left[ \mathbbm{1} \left(y_{1} +\cdots+y_{n-1}>\theta_{0}-x_n\right) \right] \\
    &\qquad \qquad +\frac{1-\rho}{2} \underset{y_{-n}\sim N_{\rho}\left(x_{-n}\right)}{ \mathbb{E}}  \left[ \mathbbm{1} \left(y_{1} +\cdots+y_{n-1}>\theta_{0}+x_n\right) \right] \\
    &=\frac{1+\rho}{2} T_{\rho} f_{\theta_{0}-x_{n}}\left(x_{-n}\right)+\frac{1-\rho}{2} T_{\rho} f_{\theta_{0}+x_{n}}\left(x_{-n}\right)
\end{align*}
To summarize, this dynamic programming with memoization algorithm is as shown \cref{fig:algorithm}. In terms of notation we denote a specific dictionary (in terms of popular programming terminology of dictionary data types) as $\text {Dictionary: }\left\{(\rho, n, s, \theta)=T_{\rho} f_{\theta_{0}}(x) \text { for some } x \text { s.t } sum(x)=s \right\}$.

 Our approach is to use this proposed recursive relation with an appropriate initial condition to exactly compute the noise operator $T_{\rho}f(x)$. Then, by using \cref{Stab-Markov}, we calculate the Stability of the function. Finally, by using the linear relation between stability and accuracy from \cref{Stab-Acc}, we compute the exact accuracy. This dynamic programming approach avoids having to make $2^n$ computations, given that $x \sim \left\{-1,1\right\}^n$. Note that, $T_{p} f_{\theta_{0}}(x)=T_{p} f_{\theta_{0}}(z) \text { if } sum(x)=sum(z)$. Therefore we iterate over $i$ from $1$ to $n$ to represent vectors with $i$ number of $1's$. Then as the rest of entries are $-1$, and since the length of the array is $n$, this approach can model the exact sum of all possible vectors. Since the calculation of the stability is one-to-one with respect to sums, we store the intermediate results in a dictionary indexed by this sum. As there are $n \choose i$ vectors that can be represented this way, we just compute once per each $i$ and multiply it by $n \choose i$. This enables us to model all possible vectors efficiently but allows us to not have to compute the intermediate results every time via our recursive approach.

\begin{figure}[h!]
    \centering
    \includegraphics[scale = 0.5]{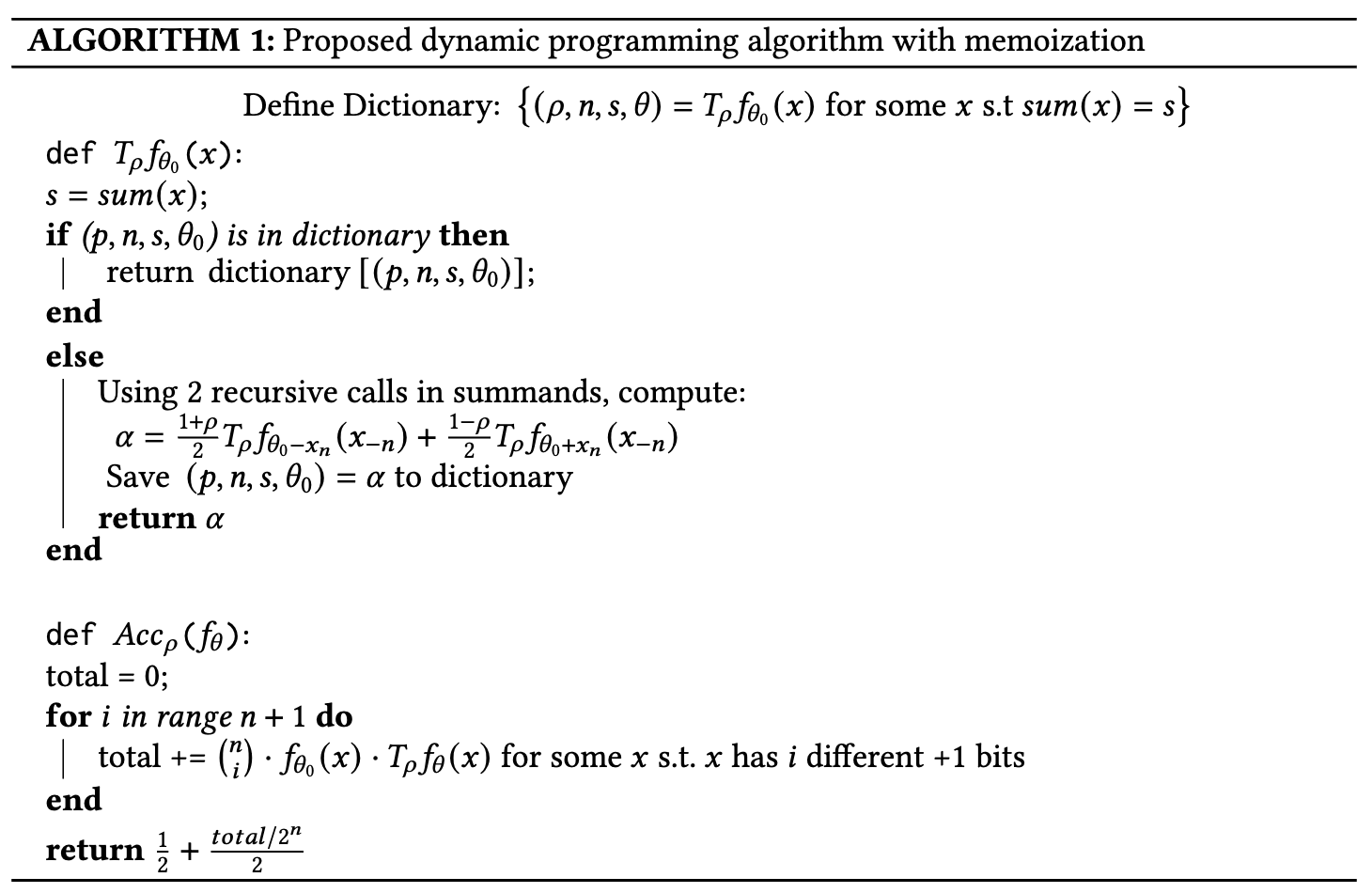}
    \caption{Proposed dynamic programming algorithm with memoization}
    \label{fig:algorithm}
\end{figure}

In \cref{fig:plot}, we plot the accuracy curves of the randomized response mechanism with varying values of $\rho$ applied to the majority function as the number of voters increases. Note that as $n$ goes to $\infty$, the accuracy asymptotically approaches to $\frac{1}{2} + \frac{1}{\pi} \arcsin(\rho)$ as implied by \cref{asymptotic accuracy of majority}.

\begin{figure}[h!]
    \centering
    \includegraphics[scale = 0.7]{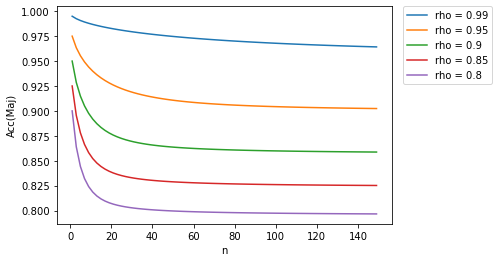}
    \caption{The accuracy curves of the randomized response mechanism with varying values of $\rho$ applied to the majority function as the number of voters increases.
}
    \label{fig:plot}
\end{figure}
\subsection{Dictatorship} Let $f: \{-1,1\}^n \to \{-1,1\}$ be the dictatorship of voter-$i$, that is $f(x) = 1$ if and only if $x_i =1$.

Then, for any given $x \in \{-1,1\}^n$, 
\[\PP[M_{\rho}f(x) = f(x)] = \PP_{y \sim N_{\rho}(x)}[f(y) = f(x)] = \PP_{y_i \sim N_{\rho}(x_i)}[y_i = x_i] = \frac{1+\rho}{2}.\]
Hence, the average accuracy is also equal to $\frac{1+\rho}{2}$. 

\subsection{$AND_n$ and $OR_n$} We will first make the calculations for $AND_n$ and the results will be analogous due to symmetry. We will make use of Fact \ref{switching lemma} in the analysis.

First, we start with a generic calculation that holds for any social choice function $f$. In the calculations in this section, our probability space is $x \sim \{-1,1\}^n, M_{\rho}f(x) \sim f(y)$ where $y \sim N_{\rho}x$.

Note that by Fact \ref{switching lemma}, 
\[\PP_{x, M_{\rho}} [M_{\rho}f(x) = 1] = \PP_x[f(x) = 1].\]

\begin{align*}
\PP\left[M_{\rho}f(x) = f(x) \right] 
=\PP\left[M_{\rho}f(x) = 1 \wedge f(x) = 1 \right] + \PP\left[M_{\rho}f(x) = -1 \wedge f(x) = -1 \right]
\end{align*}
and
\begin{align*}
\PP\left[M_{\rho}f(x) = -1 \wedge f(x) = -1 \right]
&= 1-\PP\left[M_{\rho}f(x) = 1 \vee f(x) = 1 \right] \\
&= 1-\PP\left[M_{\rho}f(x) = 1\right] - \PP\left[f(x) = 1 \right] + \PP\left[M_{\rho}f(x) = 1 \wedge f(x) = 1 \right] \\
&= 1 - 2 \cdot \PP\left[f(x) = 1 \right] + \PP\left[M_{\rho}f(x) = 1 \wedge f(x) = 1 \right].
\end{align*}
Thus for any social choice function $f$, 
\begin{align*}
\PP\left[M_{\rho}f(x) = f(x) \right] 
=1-2 \cdot \PP\left[f(x) = 1 \right] + 2 \cdot \PP\left[M_{\rho}f(x) = 1 \wedge f(x) = 1 \right]
\end{align*}

For $f = AND_n$, 
\[ \PP[f(x) = 1] = \prod_{i \in [n]} \PP[x_i = 1] = 2^{-n},\]
and
\[ \PP\left[M_{\rho}f(x) = 1 \wedge f(x) = 1 \right] = \PP[f(x) = 1] \cdot \PP[M_{\rho}f(x) = 1 | f(x) = 1] = 2^{-n} \cdot (\frac{1+\rho}{2})^{-n}.\] 
Hence, the accuracy of $M_{\rho}$ for $AND_n$ function is equal to $1 - 2^{-n+1}(1-(\frac{1+\rho}{2})^{n})$, whose limit goes to $1$ as $n$ goes to $\infty$. Due to symmetry, accuracy analysis is the same for $OR_n$ function.

\section{Conclusion}\label{conclusion}
The main objective in this work is to study the privacy-welfare trade-off and the relation between privacy and probabilistic influence.
The proposed definition of welfare happens to hold for any mechanism while on the other hand, the defined probabilistic influence is only specific to the randomized response mechanism. In fact, a more general definition of influence could be coined and a similar property could potentially be observed. We leave out this potential generalization of influence to future work. In terms of welfare, the analysis done in this paper can be replicated in a similar style to other popular privatization schemes such  as the Laplace and exponential mechanisms. The privacy-accuracy trade-off of the current mechanism for the majority function may also be further improved. Note that Dictatorship, AND, and OR functions satisfy the equality condition in \cref{privacy} as discussed in \cref{privacy remark}. Thus, the accuracy-privacy analyses for these functions are tight. On the other hand, for a given $\rho$, the asymptotic accuracy of majority is tight whereas the privacy result is a possibly loose upper bound. 

Also, our definitions of influence and welfare assume that the votes are unbiased, that is, they consider everybody to be equally likely to vote for $-1$ or $+1$. In fact, these definitions can be further generalized to cover the same concept, but for the case of biased voting. For example, one can extend the definitions to be \emph{$p$-biased} for a given $p \in [-1,1]$, that is the expected value of each vote is $p$ instead of $0$. $p$-biased distribution is also well-studied in the field of Analysis of Boolean functions. 

Finally, our voting model in this paper is a classical referendum model with two candidates. However, in most real-world applications, we generally have multiple candidates and we have to aggregate the rankings. If there is a Condorcet winner in a voting system, then the results regarding two-candidate elections can be directly applied in the multiple-candidate setting. Yet, in many cases, there is no Condorcet winner. Restricting the number of candidates to two has the primary advantage that both the definitions and analyses of welfare and influence naturally follow. We believe that extending the definitions and the tools developed in this paper to multiple-candidate settings would be interesting.

In a broader perspective, we study the effect of using privacy inducing randomized responses in the voting process. We construct a relation between the level of privacy and the resulting level of influence of voters involved in the voting system and the welfare of the chosen social choice function. An insightful takeaway that we can deduce from the derived relationships in this paper is that the ordering of voters' influences and the ordering of welfare amongst the considered social choice functions remain unchanged upon introducing noise via the celebrated randomized response mechanism. Existing works have extensively studied the relationship between privacy and the resulting accuracy in preserving the output of the query that was privatized. At a high level we are the first to shed light on the relationship between privacy and other important phenomena of influence and welfare. We hope that this bridge we have proposed between the two important fields of differential privacy and social choice theory will be further studied and extended as part of future works.




\bibliography{sn-bibliography}

\begin{thebibliography}{}
\renewcommand{\doi}[1]{\url{https://doi.org/#1}}
\bibcommenthead

\bibitem [\protect \citeauthoryear {%
Arrow%
}{%
Arrow%
}{%
{\protect \APACyear {1950}}%
}]{%
arrow1950difficulty}
\APACinsertmetastar {%
arrow1950difficulty}%
\begin{APACrefauthors}%
Arrow, K.J.%
\end{APACrefauthors}%
\unskip\
\newblock
\APACrefYearMonthDay{1950}{}{}.
\newblock
{\BBOQ}\APACrefatitle {A difficulty in the concept of social welfare} {A
  difficulty in the concept of social welfare}.{\BBCQ}
\newblock
\APACjournalVolNumPages{Journal of political economy}{58}{4}{328--346,}
\newblock

\newblock

\PrintBackRefs{\CurrentBib}

\bibitem [\protect \citeauthoryear {%
Arrow%
}{%
Arrow%
}{%
{\protect \APACyear {1951}}%
}]{%
arrow2012social}
\APACinsertmetastar {%
arrow2012social}%
\begin{APACrefauthors}%
Arrow, K.J.%
\end{APACrefauthors}%
\unskip\
\newblock
\APACrefYear{1951}.
\newblock
\APACrefbtitle {Social choice and individual values} {Social choice and
  individual values}.
\newblock
\APACaddressPublisher{}{Yale university press}.
\PrintBackRefs{\CurrentBib}

\bibitem [\protect \citeauthoryear {%
Ben-Or%
\ \BBA {} Linial%
}{%
Ben-Or%
\ \BBA {} Linial%
}{%
{\protect \APACyear {1985}}%
}]{%
influenceBOL85}
\APACinsertmetastar {%
influenceBOL85}%
\begin{APACrefauthors}%
Ben-Or, M.%
\BCBT {}\ \BBA {} Linial, N.%
\end{APACrefauthors}%
\unskip\
\newblock
\APACrefYearMonthDay{1985}{}{}.
\newblock
{\BBOQ}\APACrefatitle {Collective coin flipping, robust voting schemes and
  minima of Banzhaf values} {Collective coin flipping, robust voting schemes
  and minima of banzhaf values}.{\BBCQ}
\newblock
\APACjournalVolNumPages{Proceedings of the 26th Annual IEEE Symposium on
  Foundations of Computer Science}{}{}{408–416,}
\newblock

\newblock

\PrintBackRefs{\CurrentBib}

\bibitem [\protect \citeauthoryear {%
Cummings%
, Gupta%
, Kimpara%
\BCBL {}\ \BBA {} Morgenstern%
}{%
Cummings%
\ \protect \BOthers {.}}{%
{\protect \APACyear {2019}}%
}]{%
cummings2019compatibility}
\APACinsertmetastar {%
cummings2019compatibility}%
\begin{APACrefauthors}%
Cummings, R.%
, Gupta, V.%
, Kimpara, D.%
\BCBL {} Morgenstern, J.%
\end{APACrefauthors}%
\unskip\
\newblock
\APACrefYearMonthDay{2019}{}{}.
\newblock
{\BBOQ}\APACrefatitle {On the compatibility of privacy and fairness} {On the
  compatibility of privacy and fairness}.{\BBCQ}
\newblock
 \APACrefbtitle {Adjunct Publication of the 27th Conference on User Modeling,
  Adaptation and Personalization} {Adjunct publication of the 27th conference
  on user modeling, adaptation and personalization}\ (\BPGS\ 309--315).
\PrintBackRefs{\CurrentBib}

\bibitem [\protect \citeauthoryear {%
Dwork%
\ \BBA {} Lei%
}{%
Dwork%
\ \BBA {} Lei%
}{%
{\protect \APACyear {2009}}%
}]{%
dwork2009differential}
\APACinsertmetastar {%
dwork2009differential}%
\begin{APACrefauthors}%
Dwork, C.%
\BCBT {}\ \BBA {} Lei, J.%
\end{APACrefauthors}%
\unskip\
\newblock
\APACrefYearMonthDay{2009}{}{}.
\newblock
{\BBOQ}\APACrefatitle {Differential privacy and robust statistics}
  {Differential privacy and robust statistics}.{\BBCQ}
\newblock
 \APACrefbtitle {Proceedings of the forty-first annual ACM symposium on Theory
  of computing} {Proceedings of the forty-first annual acm symposium on theory
  of computing}\ (\BPGS\ 371--380).
\PrintBackRefs{\CurrentBib}

\bibitem [\protect \citeauthoryear {%
Dwork%
, McSherry%
, Nissim%
\BCBL {}\ \BBA {} Smith%
}{%
Dwork%
\ \protect \BOthers {.}}{%
{\protect \APACyear {2006}}%
}]{%
DMNS06}
\APACinsertmetastar {%
DMNS06}%
\begin{APACrefauthors}%
Dwork, C.%
, McSherry, F.%
, Nissim, K.%
\BCBL {} Smith, A.%
\end{APACrefauthors}%
\unskip\
\newblock
\APACrefYearMonthDay{2006}{}{}.
\newblock
{\BBOQ}\APACrefatitle {Calibrating Noise to Sensitivity in Private Data
  Analysis} {Calibrating noise to sensitivity in private data analysis}.{\BBCQ}
\newblock
 S.~Halevi\ \BBA {} T.~Rabin\ (\BEDS), \APACrefbtitle {Theory of Cryptography}
  {Theory of cryptography}\ (\BPGS\ 265--284).
\newblock
\APACaddressPublisher{Berlin, Heidelberg}{Springer Berlin Heidelberg}.
\PrintBackRefs{\CurrentBib}

\bibitem [\protect \citeauthoryear {%
Dwork%
\ \BBA {} Roth%
}{%
Dwork%
\ \BBA {} Roth%
}{%
{\protect \APACyear {2014}}%
}]{%
dwork2014algorithmic}
\APACinsertmetastar {%
dwork2014algorithmic}%
\begin{APACrefauthors}%
Dwork, C.%
\BCBT {}\ \BBA {} Roth, A.%
\end{APACrefauthors}%
\unskip\
\newblock
\APACrefYearMonthDay{2014}{}{}.
\newblock
{\BBOQ}\APACrefatitle {The algorithmic foundations of differential privacy.}
  {The algorithmic foundations of differential privacy.}{\BBCQ}
\newblock
\APACjournalVolNumPages{Found. Trends Theor. Comput. Sci.}{9}{3-4}{211--407,}
\newblock

\newblock

\PrintBackRefs{\CurrentBib}

\bibitem [\protect \citeauthoryear {%
Erlingsson%
, Pihur%
\BCBL {}\ \BBA {} Korolova%
}{%
Erlingsson%
\ \protect \BOthers {.}}{%
{\protect \APACyear {2014}}%
}]{%
erlingsson2014rappor}
\APACinsertmetastar {%
erlingsson2014rappor}%
\begin{APACrefauthors}%
Erlingsson, {\'U}.%
, Pihur, V.%
\BCBL {} Korolova, A.%
\end{APACrefauthors}%
\unskip\
\newblock
\APACrefYearMonthDay{2014}{}{}.
\newblock
{\BBOQ}\APACrefatitle {Rappor: Randomized aggregatable privacy-preserving
  ordinal response} {Rappor: Randomized aggregatable privacy-preserving ordinal
  response}.{\BBCQ}
\newblock
 \APACrefbtitle {Proceedings of the 2014 ACM SIGSAC conference on computer and
  communications security} {Proceedings of the 2014 acm sigsac conference on
  computer and communications security}\ (\BPGS\ 1054--1067).
\PrintBackRefs{\CurrentBib}

\bibitem [\protect \citeauthoryear {%
Garman%
\ \BBA {} Kamien%
}{%
Garman%
\ \BBA {} Kamien%
}{%
{\protect \APACyear {1968}}%
}]{%
GK68}
\APACinsertmetastar {%
GK68}%
\begin{APACrefauthors}%
Garman, M.B.%
\BCBT {}\ \BBA {} Kamien, M.I.%
\end{APACrefauthors}%
\unskip\
\newblock
\APACrefYearMonthDay{1968}{}{}.
\newblock
{\BBOQ}\APACrefatitle {The paradox of voting: Probability calculations} {The
  paradox of voting: Probability calculations}.{\BBCQ}
\newblock
\APACjournalVolNumPages{Behavioral Science}{13}{4}{306-316,}
\newblock

\newblock

\PrintBackRefs{\CurrentBib}

\bibitem [\protect \citeauthoryear {%
Gibbard%
}{%
Gibbard%
}{%
{\protect \APACyear {1973}}%
}]{%
GS73}
\APACinsertmetastar {%
GS73}%
\begin{APACrefauthors}%
Gibbard, A.%
\end{APACrefauthors}%
\unskip\
\newblock
\APACrefYearMonthDay{1973}{}{}.
\newblock
{\BBOQ}\APACrefatitle {Manipulation of Voting Schemes: A General Result}
  {Manipulation of voting schemes: A general result}.{\BBCQ}
\newblock
\APACjournalVolNumPages{Econometrica}{41}{4}{587--601,}
\newblock

\newblock

\PrintBackRefs{\CurrentBib}

\bibitem [\protect \citeauthoryear {%
Guilbaud%
}{%
Guilbaud%
}{%
{\protect \APACyear {2012}}%
}]{%
Guildbau}
\APACinsertmetastar {%
Guildbau}%
\begin{APACrefauthors}%
Guilbaud, G\BHBI T.%
\end{APACrefauthors}%
\unskip\
\newblock
\APACrefYearMonthDay{2012}{}{}.
\newblock
{\BBOQ}\APACrefatitle {Les théories de l'intérêt général et le problème
  logique de l'agrégation} {Les théories de l'intérêt général et le
  problème logique de l'agrégation}.{\BBCQ}
\newblock
\APACjournalVolNumPages{Revue économique}{63}{4}{659-720,}
\newblock

\newblock

\PrintBackRefs{\CurrentBib}

\bibitem [\protect \citeauthoryear {%
Hay%
, Elagina%
\BCBL {}\ \BBA {} Miklau%
}{%
Hay%
\ \protect \BOthers {.}}{%
{\protect \APACyear {2017}}%
}]{%
DPRankAggregation2017}
\APACinsertmetastar {%
DPRankAggregation2017}%
\begin{APACrefauthors}%
Hay, M.%
, Elagina, L.%
\BCBL {} Miklau, G.%
\end{APACrefauthors}%
\unskip\
\newblock
\APACrefYearMonthDay{2017}{}{}.
\newblock
{\BBOQ}\APACrefatitle {Differentially Private Rank Aggregation} {Differentially
  private rank aggregation}.{\BBCQ}
\newblock
\BIn{} \APACrefbtitle {Proceedings of the 2017 SIAM International Conference on
  Data Mining (SDM)} {Proceedings of the 2017 siam international conference on
  data mining (sdm)}\ (\BPG~669-677).
\PrintBackRefs{\CurrentBib}

\bibitem [\protect \citeauthoryear {%
Hillinger%
}{%
Hillinger%
}{%
{\protect \APACyear {2005}}%
}]{%
Hillinger2005}
\APACinsertmetastar {%
Hillinger2005}%
\begin{APACrefauthors}%
Hillinger, C.%
\end{APACrefauthors}%
\unskip\
\newblock
\APACrefYearMonthDay{2005}{}{}.
\newblock
{\BBOQ}\APACrefatitle {The Case for Utilitarian Voting} {The case for
  utilitarian voting}.{\BBCQ}
\newblock
\APACjournalVolNumPages{Homo Oeconomicus}{22}{3}{295-321,}
\newblock

\newblock

\PrintBackRefs{\CurrentBib}

\bibitem [\protect \citeauthoryear {%
Kahn%
, Kalai%
\BCBL {}\ \BBA {} Linial%
}{%
Kahn%
\ \protect \BOthers {.}}{%
{\protect \APACyear {1988}}%
}]{%
KKL}
\APACinsertmetastar {%
KKL}%
\begin{APACrefauthors}%
Kahn, J.%
, Kalai, G.%
\BCBL {} Linial, N.%
\end{APACrefauthors}%
\unskip\
\newblock
\APACrefYearMonthDay{1988}{}{}.
\newblock
{\BBOQ}\APACrefatitle {The influence of variables on Boolean functions} {The
  influence of variables on boolean functions}.{\BBCQ}
\newblock
 \APACrefbtitle {[Proceedings 1988] 29th Annual Symposium on Foundations of
  Computer Science} {[proceedings 1988] 29th annual symposium on foundations of
  computer science}\ (\BPG~68-80).
\PrintBackRefs{\CurrentBib}

\bibitem [\protect \citeauthoryear {%
Lee%
}{%
Lee%
}{%
{\protect \APACyear {2015}}%
}]{%
lee2015efficient}
\APACinsertmetastar {%
lee2015efficient}%
\begin{APACrefauthors}%
Lee, D.T.%
\end{APACrefauthors}%
\unskip\
\newblock
\APACrefYearMonthDay{2015}{}{}.
\newblock
{\BBOQ}\APACrefatitle {Efficient, Private, and eps-Strategyproof Elicitation of
  Tournament Voting Rules} {Efficient, private, and eps-strategyproof
  elicitation of tournament voting rules}.{\BBCQ}
\newblock
 \APACrefbtitle {Twenty-Fourth International Joint Conference on Artificial
  Intelligence.} {Twenty-fourth international joint conference on artificial
  intelligence.}
\PrintBackRefs{\CurrentBib}

\bibitem [\protect \citeauthoryear {%
Mandal%
, Procaccia%
, Shah%
\BCBL {}\ \BBA {} Woodruff%
}{%
Mandal%
\ \protect \BOthers {.}}{%
{\protect \APACyear {2019}}%
}]{%
Mandal2019}
\APACinsertmetastar {%
Mandal2019}%
\begin{APACrefauthors}%
Mandal, D.%
, Procaccia, A.D.%
, Shah, N.%
\BCBL {} Woodruff, D.%
\end{APACrefauthors}%
\unskip\
\newblock
\APACrefYearMonthDay{2019}{}{}.
\newblock
{\BBOQ}\APACrefatitle {Efficient and Thrifty Voting by Any Means Necessary}
  {Efficient and thrifty voting by any means necessary}.{\BBCQ}
\newblock
 H.~Wallach, H.~Larochelle, A.~Beygelzimer, F.~d\textquotesingle Alch\'{e}-Buc,
  E.~Fox\BCBL {}\ \BBA {} R.~Garnett\ (\BEDS), \APACrefbtitle {Advances in
  Neural Information Processing Systems} {Advances in neural information
  processing systems}\ (\BVOL~32).
\newblock
\APACaddressPublisher{}{Curran Associates, Inc.}
\PrintBackRefs{\CurrentBib}

\bibitem [\protect \citeauthoryear {%
Mandal%
, Shah%
\BCBL {}\ \BBA {} Woodruff%
}{%
Mandal%
\ \protect \BOthers {.}}{%
{\protect \APACyear {2020}}%
}]{%
Mandal2020}
\APACinsertmetastar {%
Mandal2020}%
\begin{APACrefauthors}%
Mandal, D.%
, Shah, N.%
\BCBL {} Woodruff, D.P.%
\end{APACrefauthors}%
\unskip\
\newblock
\APACrefYearMonthDay{2020}{}{}.
\newblock
{\BBOQ}\APACrefatitle {Optimal Communication-Distortion Tradeoff in Voting}
  {Optimal communication-distortion tradeoff in voting}.{\BBCQ}
\newblock
 \APACrefbtitle {Proceedings of the 21st ACM Conference on Economics and
  Computation} {Proceedings of the 21st acm conference on economics and
  computation}\ (\BPG~795–813).
\newblock
\APACaddressPublisher{New York, NY, USA}{Association for Computing Machinery}.
\PrintBackRefs{\CurrentBib}

\bibitem [\protect \citeauthoryear {%
Mangat%
}{%
Mangat%
}{%
{\protect \APACyear {1994}}%
}]{%
mangat1994improved}
\APACinsertmetastar {%
mangat1994improved}%
\begin{APACrefauthors}%
Mangat, N.S.%
\end{APACrefauthors}%
\unskip\
\newblock
\APACrefYearMonthDay{1994}{}{}.
\newblock
{\BBOQ}\APACrefatitle {An improved randomized response strategy} {An improved
  randomized response strategy}.{\BBCQ}
\newblock
\APACjournalVolNumPages{Journal of the Royal Statistical Society: Series B
  (Methodological)}{56}{1}{93--95,}
\newblock

\newblock

\PrintBackRefs{\CurrentBib}

\bibitem [\protect \citeauthoryear {%
May%
}{%
May%
}{%
{\protect \APACyear {1952}}%
}]{%
May1952}
\APACinsertmetastar {%
May1952}%
\begin{APACrefauthors}%
May, K.O.%
\end{APACrefauthors}%
\unskip\
\newblock
\APACrefYearMonthDay{1952}{}{}.
\newblock
{\BBOQ}\APACrefatitle {A set of independent necessary and sufficient conditions
  for simple majority decisions} {A set of independent necessary and sufficient
  conditions for simple majority decisions}.{\BBCQ}
\newblock
\APACjournalVolNumPages{Econometrica}{20}{4}{680-684,}
\newblock

\newblock

\PrintBackRefs{\CurrentBib}

\bibitem [\protect \citeauthoryear {%
Mossel%
, O'Donnell%
\BCBL {}\ \BBA {} Oleszkiewicz%
}{%
Mossel%
\ \protect \BOthers {.}}{%
{\protect \APACyear {2010}}%
}]{%
MajorityIsStablest2010}
\APACinsertmetastar {%
MajorityIsStablest2010}%
\begin{APACrefauthors}%
Mossel, E.%
, O'Donnell, R.%
\BCBL {} Oleszkiewicz, K.%
\end{APACrefauthors}%
\unskip\
\newblock
\APACrefYearMonthDay{2010}{}{}.
\newblock
{\BBOQ}\APACrefatitle {Noise stability of functions with low influences:
  Invariance and optimality} {Noise stability of functions with low influences:
  Invariance and optimality}.{\BBCQ}
\newblock
\APACjournalVolNumPages{Annals of Mathematics}{171}{1}{295–341,}
\newblock

\newblock

\PrintBackRefs{\CurrentBib}

\bibitem [\protect \citeauthoryear {%
O'Donnell%
}{%
O'Donnell%
}{%
{\protect \APACyear {2004}}%
}]{%
ODonnel2004}
\APACinsertmetastar {%
ODonnel2004}%
\begin{APACrefauthors}%
O'Donnell, R.%
\end{APACrefauthors}%
\unskip\
\newblock
\APACrefYearMonthDay{2004}{}{}.
\newblock
{\BBOQ}\APACrefatitle {Hardness amplification within NP} {Hardness
  amplification within np}.{\BBCQ}
\newblock
\APACjournalVolNumPages{Journal of Computer and System Sciences}{69}{1}{68-94,}
\newblock
\APACrefnote{Special Issue on Computational Complexity 2002}
\newblock

\newblock

\PrintBackRefs{\CurrentBib}

\bibitem [\protect \citeauthoryear {%
O'Donnell%
}{%
O'Donnell%
}{%
{\protect \APACyear {2014}}%
}]{%
o2014analysis}
\APACinsertmetastar {%
o2014analysis}%
\begin{APACrefauthors}%
O'Donnell, R.%
\end{APACrefauthors}%
\unskip\
\newblock
\APACrefYear{2014}.
\newblock
\APACrefbtitle {Analysis of boolean functions} {Analysis of boolean functions}.
\newblock
\APACaddressPublisher{}{Cambridge University Press}.
\PrintBackRefs{\CurrentBib}

\bibitem [\protect \citeauthoryear {%
Penrose%
}{%
Penrose%
}{%
{\protect \APACyear {1946}}%
}]{%
penrose46}
\APACinsertmetastar {%
penrose46}%
\begin{APACrefauthors}%
Penrose, L.%
\end{APACrefauthors}%
\unskip\
\newblock
\APACrefYearMonthDay{1946}{}{}.
\newblock
{\BBOQ}\APACrefatitle {The elementary statistics of majority voting} {The
  elementary statistics of majority voting}.{\BBCQ}
\newblock
\APACjournalVolNumPages{Journal of the Royal Statistical
  Society}{}{}{109(1):53–57,}
\newblock

\newblock

\PrintBackRefs{\CurrentBib}

\bibitem [\protect \citeauthoryear {%
Rousseau%
}{%
Rousseau%
}{%
{\protect \APACyear {1762}}%
}]{%
rousseau}
\APACinsertmetastar {%
rousseau}%
\begin{APACrefauthors}%
Rousseau, J\BHBI J.%
\end{APACrefauthors}%
\unskip\
\newblock
\APACrefYearMonthDay{1762}{}{}.
\newblock
{\BBOQ}\APACrefatitle {Du Contrat Social} {Du contrat social}.{\BBCQ}
\newblock
\APACjournalVolNumPages{Marc-Michel Rey}{}{}{,}
\newblock

\newblock

\PrintBackRefs{\CurrentBib}

\bibitem [\protect \citeauthoryear {%
Schwartzberg%
}{%
Schwartzberg%
}{%
{\protect \APACyear {2008}}%
}]{%
S08}
\APACinsertmetastar {%
S08}%
\begin{APACrefauthors}%
Schwartzberg, M.%
\end{APACrefauthors}%
\unskip\
\newblock
\APACrefYearMonthDay{2008}{}{}.
\newblock
{\BBOQ}\APACrefatitle {Voting the General Will: Rousseau on Decision Rules}
  {Voting the general will: Rousseau on decision rules}.{\BBCQ}
\newblock
\APACjournalVolNumPages{Political Theory}{36}{3}{403--423,}
\newblock
\begin{APACrefURL} [{2023-05-04}]{http://www.jstor.org/stable/20452639}
  \end{APACrefURL}
\newblock

\newblock

\PrintBackRefs{\CurrentBib}

\bibitem [\protect \citeauthoryear {%
Shang%
, Wang%
, Cuff%
\BCBL {}\ \BBA {} Kulkarni%
}{%
Shang%
\ \protect \BOthers {.}}{%
{\protect \APACyear {2014}}%
}]{%
shang2014}
\APACinsertmetastar {%
shang2014}%
\begin{APACrefauthors}%
Shang, S.%
, Wang, T.%
, Cuff, P.%
\BCBL {} Kulkarni, S.%
\end{APACrefauthors}%
\unskip\
\newblock
\APACrefYearMonthDay{2014}{}{}.
\newblock
\APACrefbtitle {The Application of Differential Privacy for Rank Aggregation:
  Privacy and Accuracy.} {The application of differential privacy for rank
  aggregation: Privacy and accuracy.}
\PrintBackRefs{\CurrentBib}

\bibitem [\protect \citeauthoryear {%
Titsworth%
}{%
Titsworth%
}{%
{\protect \APACyear {1962}}%
}]{%
Titsworth}
\APACinsertmetastar {%
Titsworth}%
\begin{APACrefauthors}%
Titsworth, R.%
\end{APACrefauthors}%
\unskip\
\newblock
\APACrefYearMonthDay{1962}{}{}.
\newblock
{\BBOQ}\APACrefatitle {Correlation properties of cyclic sequences} {Correlation
  properties of cyclic sequences}.{\BBCQ}
\newblock
\APACjournalVolNumPages{PhD thesis, CalTech}{}{}{,}
\newblock

\newblock

\PrintBackRefs{\CurrentBib}

\bibitem [\protect \citeauthoryear {%
Warner%
}{%
Warner%
}{%
{\protect \APACyear {1965}}%
}]{%
warner1965randomized}
\APACinsertmetastar {%
warner1965randomized}%
\begin{APACrefauthors}%
Warner, S.L.%
\end{APACrefauthors}%
\unskip\
\newblock
\APACrefYearMonthDay{1965}{}{}.
\newblock
{\BBOQ}\APACrefatitle {Randomized response: A survey technique for eliminating
  evasive answer bias} {Randomized response: A survey technique for eliminating
  evasive answer bias}.{\BBCQ}
\newblock
\APACjournalVolNumPages{Journal of the American Statistical
  Association}{60}{309}{63--69,}
\newblock

\newblock

\PrintBackRefs{\CurrentBib}

\end{thebibliography}

\appendix 
\section{Proofs} \label{appendix: proofs}

\subsection{Proof of Proposition \ref{privacy}} \label{proof of RR}
\begin{proof}
Let $r$ be any element in the range of $M_{\rho}f$. 
Let $Z = \{z \in \{-1,1\}^n| f(z) =r\}$. Let $x$ and $x'$ differ only at $x_i$ for some $i \in [n]$. 
\[ \frac{\PP [M_{\rho}f(x) =r]}{\PP [M_{\rho}f(x') = r]}
 = \frac{\sum_{z \in Z} \PP_{y \sim N_{\rho}x}[y = z]}{\sum_{z \in Z} \PP_{y \sim N_{\rho}x'}[y = z]}
 = \frac{\sum_{z \in Z} \prod_{j \in [n]} \PP_{y_j \sim N_{\rho}x_j}[y_j = z_j]}{\sum_{z \in Z} \prod_{j \in [n]} \PP_{y_j \sim N_{\rho}x'_j}[y_j = z_j]}.
\]
The first equality is upon considering all cases of output of the randomized response resulting in a $z\in Z$. Then by definition that would result in the function $f$ evaluated on this output $z$ to be $r$. The second equality is due to the independence assumption across the voters choices.\\ 
Now, for any $z \in Z$,
\[ \PP_{y_j \sim N_{\rho}x_j}[y_j = z_j] =
\begin{cases}
\frac{1+\rho}{2} & \text{if } x_j = z_j \\
\frac{1-\rho}{2} & \text{if } x_j \neq z_j
\end{cases}
\quad \text{and} \quad
\PP_{y_j \sim N_{\rho}x'_j}[y_j = z_j] =
\begin{cases}
\frac{1+\rho}{2} & \text{if } x'_j = z_j \\
\frac{1-\rho}{2} & \text{if } x'_j \neq z_j
\end{cases}
\]
This is because $\frac{1-\rho}{2}$ is the probability of a misrecorded vote and $1-\frac{1-\rho}{2}=\frac{1+\rho}{2}$ is the probability otherwise. More explicitly, with probability $1-\rho$, it chooses to blur the ballot and the blurring is then done by picking uniformly out of the two options of $\{-1,1\}$ with probability $0.5$ each, out of which one pick would result in no change to the vote and the other would result in a misrecorded vote.
Also, for any $j \neq i$, 
\[\PP_{y_j \sim N_{\rho}x_j}[y_j = z_j] = \PP_{y_j \sim N_{\rho}x'_j}[y_j = z_j].\]
Thus, 
 \[
 \frac{1-\rho}{1+\rho}
 \leq
 \frac{\sum_{z \in Z} \prod_{j \in [n]} \PP_{y_j \sim N_{\rho}x_j}[y_j = z_j]}{\sum_{z \in Z} \prod_{j \in [n]} \PP_{y_j \sim N_{\rho}x'_j}[y_j = z_j]}
 \leq 
 \frac{1+\rho}{1-\rho},\]
 which completes the proof.
\end{proof}
%

\subsection{Proof of \cref{influence theorem}} \label{appendix proof of influence theorem}

\begin{proof}  Using conditional probability, we get that
\begin{align*}
I_i[M_{\rho}f] 
&= \EE_{x \sim \{-1,1\}^n, \forall j \neq i \text{ } z_j=y_j=x_j, y_i \sim N_{\rho}(1), z_i \sim N_{\rho}(-1)}\left[\left(\frac{f(y) - f(z)}{2}\right)^2\right] \\
&= \PP_{ y_i \sim N_{\rho}(1), z_i \sim N_{\rho}(-1)}[y_i =1, z_i = -1] \cdot \EE_{x \sim \{-1,1\}^n}\left[\left( \frac{f(x_{i \to 1}) - f(x_{i \to -1})}{2} \right)^2\right] \\
&+ \PP_{y_i \sim N_{\rho}(1), z_i \sim N_{\rho}(-1)}[y_i =-1, z_i = 1] \cdot \EE_{x \sim \{-1,1\}^n}\left[\left(\frac{f(x_{i \to 1}) - f(x_{i \to -1})}{2}\right)^2\right] \\
\end{align*}
Noting that 
\[ \PP_{y_i \sim N_{\rho}(1), z_i \sim N_{\rho}(-1)}\left[y_i =1, z_i = -1\right] = \left(\frac{1 + \rho}{2}\right)^2,\] 
\[\PP_{y_i \sim N_{\rho}(1), z_i \sim N_{\rho}(-1)}\left[y_i =-1, z_i = 1\right] = \left(\frac{1 - \rho}{2}\right)^2,\]
and that
\[ \EE_{x \sim \{-1,1\}^n}\left[\left( \frac{f(x_{i \to 1}) - f(x_{i \to -1})}{2} \right)^2\right] = I_i[f],\]
we get that 
\[I_i[M_{\rho}f] =\frac{1 + \rho^2}{2} I_i[f].\]
\end{proof}

\subsection{Proof of \cref{majority is welfare maximizer}} \label{proof majority welfare maximizer}

\begin{proof}
First, let us fix $x$. Note that 
\[w_x(f) = f(x)\cdot \sum_{i \in [n]} x_i.\] 
Since $f(x) \in \{-1,1\},$ $f(x)\cdot \sum_{i \in [n]} x_i$ is maximized when $f(x) = sign(\sum_{i \in [n]} x_i)$. Hence, $W(f)$ is maximized if $\forall x \in \{-1,1\}^n$,  $f(x) = sign(\sum_{i \in [n]} x_i)$, which is exactly the definition of the majority function.
\end{proof} 

\begin{remark}
Note that we used the condition that $n$ is odd to ensure that $sign$ function is well-defined. If $n$ was even, then the maximizers of $W(f)$ are again the majority functions where it does not matter who is elected if it is tied.
\end{remark}

\subsection{Proof of \cref{W = I}}\label{proof of W = I}

In the proof of this result, we use discrete Fourier analysis. It is a well-known result from the field of analysis of Boolean functions, that every function $f:\{-1,1\}^{n} \rightarrow \mathbb{R}$ can be uniquely expressed as a multilinear polynomial,
$$
f(x)=\sum_{S \subseteq[n]} \widehat{f}(S) \chi_{S}(x)
$$
where for any $S \in [n]$
$$
\chi_{S}(x)=\prod_{i \in S} x_{i}.
$$

This expression is called the Fourier expansion of $f$, and the real number $\widehat{f}(S)$ is called the Fourier coefficient of $f$ on $S$. Collectively, the coefficients are called the Fourier spectrum of $f$. The following is an essential result from discrete Fourier Analysis. \\

\begin{lemma}[Plancherel's Theorem]\label{Plancherel} For any functions $f,g : \{-1,1\}^n \to \RR$,
\[ E_{x \sim \{-1,1\}^n}[f(x)g(x)] = \sum_{S \subseteq [n]} \wh{f}(S)\wh{g}(S).\]
\end{lemma}

It is possible to neatly calculate many features of $f$ including the influences in terms of Fourier coefficients. \\

\begin{lemma}[Proposition 2.21, \cite{o2014analysis}] \label{monotone influence fourier} 
Let $f: \{-1,1\}^n \to \{-1,1\}$ be a monotone function and let the Fourier spectrum of $f$ be $f(x)=\sum_{S \subseteq[n]} \widehat{f}(S) \chi_{S}(x)$. Then, for any $i \in [n]$, 
\[I_i[f] = \wh{f}(\{i\}).\] 
\end{lemma}

It is also possible to calculate the welfare in terms of the Fourier coefficients by taking one step further from the proof of \cref{majority is welfare maximizer}. \\

\begin{lemma} \label{Welfare Fourier}
Let $f$ be any social choice function $f:\{-1,1\}^n \to \{-1,1\}$. Then, $W(f) = \sum_{i \in [n]} \wh{f}(\{i\})$.
\end{lemma}

\begin{proof} 
    By the definition of welfare,
\[W(f) = \EE_x[w_x(f)] = \EE_x[f(x) \cdot \sum_{i \in [n]} x_i ]=  \sum_{i \in [n]} \wh{f}(\{i\})\]
where the last equation follows from \cref{Plancherel}.
\end{proof}

We are ready to finish the proof.

\begin{proof}[Proof of \cref{W = I}]
    The proof follows immediately from \cref{monotone influence fourier} and \cref{Welfare Fourier}.
\end{proof}
\subsection{Proof of \cref{welfare theorem}} \label{proof of welfare theorem}

\begin{proof} We prove this identity by using a double-counting method and linearity of expectation. Fix $f$. For any $i \in [n]$, let $1_{i,x, \rho}$ be the indicator random variable defined as follows:
\[ 1_{i,x, \rho} =
\begin{cases}
1 & \text{if } M_{\rho}f(x) = x_i \\
-1 & \text{if } M_{\rho}f(x) \neq x_i
\end{cases}
\]
where the randomization is due to the randomized response. Note then when $x$ is given and $\rho=1$, there is no randomization because $M_{\rho}f(x) = f(x)$ with probability $1$. Therefore, $1_{i,x,1}$ is a deterministic function. For the sake of simplicity, we will abuse the notation and write $1_{i,x}$ instead of $1_{i,x,1}$ in the deterministic case. Then, 
\[w_x(M_{\rho}f) = \sum_{i \in [n]} 1_{i,x,\rho} \quad \text{and} \quad w_x(f) = \sum_{i \in [n]} 1_{i,x}.\]
Thus, 
\[ W(M_{\rho}f) = \EE_{M_{\rho},x}[w_x(f)] =  \EE_{x,M_{\rho}}[\sum_{i \in [n]} 1_{i,x,\rho}] = \sum_{i \in [n]}\EE_{x,M_{\rho}}[ 1_{i,x,\rho}]\]
and so 
\[ W(f) = \sum_{i \in [n]}\EE_{x}[ 1_{i,x}].\]
Now, we will show that for any $i \in [n]$,
\[ \EE_{x,M_{\rho}}[ 1_{i,x,\rho}] = \rho \cdot \EE_{x}[ 1_{i,x}].\]
First, note that
\[\EE_{x,M_{\rho}}[ 1_{i,x,\rho}] = \PP_{\substack{x \sim \{-1,1\}^n \\ y \sim N_{\rho}x}}[f(y) = x_i] - \PP_{\substack{x \sim \{-1,1\}^n \\ y \sim N_{\rho}x}}[f(y) \neq x_i]. \]
By using 
\[\PP_{\substack{x \sim \{-1,1\}^n \\ y \sim N_{\rho}x}}[f(y) = x_i] + \PP_{\substack{x \sim \{-1,1\}^n \\ y \sim N_{\rho}x}}[f(y) \neq x_i] = 1,\]
we get that 
\[ \EE_{x,M_{\rho}}[ 1_{i,x,\rho}]=2\cdot\PP_{\substack{x \sim \{-1,1\}^n \\ y \sim N_{\rho}x}}[f(y) = x_i] - 1.\]
By Fact \ref{switching lemma}, we can replace $x \sim \{-1,1\}^n, y \sim N_{\rho}x$ with $y \sim \{-1,1\}^n, x \sim N_{\rho}y$. Thus, by using conditional probability,
\begin{align*}
    \EE_{x,M_{\rho}}[1_{i,x, \rho}]
    &=2\cdot\PP_{\substack{y \sim \{-1,1\}^n \\ x \sim N_{\rho}y}}[f(y) = x_i] - 1 \\
    &= 2(\PP_{x \sim N_{\rho}y}[x_i = y_i] \cdot \PP_{\substack{y \sim \{-1,1\}^n }}[f(y) = y_i] + \PP_{x \sim N_{\rho}y}[x_i = -y_i] \cdot \PP_{\substack{y \sim \{-1,1\}^n }}[f(y) = -y_i]) -1\\
    &= (1+\rho)\cdot\PP_{\substack{y \sim \{-1,1\}^n }}[f(y) = y_i] + (1-\rho)\cdot\PP_{\substack{y \sim \{-1,1\}^n }}[f(y) \neq y_i] -1 \\
    &= \rho \cdot (\PP_{\substack{y \sim \{-1,1\}^n }}[f(y) = y_i] - \PP_{\substack{y \sim \{-1,1\}^n }}[f(y) \neq y_i]) \\
    &= \rho \cdot \EE_{x}[1_{i,x}]
\end{align*}
which completes the proof.
\end{proof}

\section{Social Choice Functions} \label{social choice theory preliminaries}
In this paper, we exclusively focus on social choice functions with two alternatives. There are many ways to interpret these functions. It can be considered as a two-candidate election or as a referendum in the context of political science. It can also be interpreted as a classifier in the context of Machine Learning. In this paper, we will generally give the interpretations in the context of two-candidate elections. 

In general, we work with the Boolean functions defined as $f: \{-1,1\}^n \to \RR$, and we denote the bit $i$ of the input $x$ by $x_i$ for any $i \in [n]$. However, we define welfare only for \emph{social choice functions}, that is the Boolean functions whose ranges are $\{-1,1\}$. We analyze accuracy only for the following specific social choice functions.

\begin{itemize}
    \item \textbf{Majority:} Suppose that $n$ is an odd number. The majority function of $n$ agents/voters is denoted by $Maj_n$ and defined as 
    \[ f(x) = sign(\sum_{i \in [n]} x_i)\]
    for any $x \in \{-1,1\}^n$ where $sign: \RR \to \{-1,0,1\}$ is the function such that 
    \[sign(a) = \frac{a}{|a|}\]
    for any $a \in \RR, a \neq 0$ and sign(0)=0.
    
    \item \textbf{Dictatorship:} For a given number $n$ and $i \in [n]$, the dictatorship of voter-$i$ is defined as 
    \[f(x) = x_i\]
    for any $x \in \{-1,1\}^n$.
    
    \item \textbf{AND$_n$:} The $AND_n$ function outputs $1$ if there is unanimity on $1$, outputs $-1$ otherwise. Namely,
    \[ f(x) = 
    \begin{cases}
    1 & \text{if } \forall i \in [n], x_i =1 \\
    -1 & \text{otherwise}
    \end{cases}\]
    
    \item \textbf{OR$_n$:} The $OR_n$ function outputs $1$ if at least one voter votes for $1$, and outputs $-1$ otherwise. In other words, it outputs $-1$ if there is unanimity on $-1$, outputs $1$ otherwise. Namely,
    \[ f(x) = 
    \begin{cases}
    -1 & \text{if } \forall i \in [n], x_i =-1 \\
    1 & \text{otherwise}
    \end{cases}\]
    
\end{itemize}

Note that, in this paper, we assume \emph{the impartial culture assumption}, that is the voters are not affected by each other and they vote independently uniform at random between two candidates.

\end{document}